\documentclass[a4paper]{article}

\usepackage[margin=1.4in]{geometry}

\usepackage[utf8]{inputenc}

\usepackage{comment}
\usepackage{subcaption}
\usepackage{graphicx}
\graphicspath{{figures/}}

\usepackage{booktabs}
\usepackage{hyperref}
\usepackage[numbers,sort&compress]{natbib}
\usepackage{enumerate}
\usepackage{cleveref}

\usepackage[dvipsnames]{xcolor}
\usepackage{framed}

\usepackage[vlined, ruled, linesnumbered]{algorithm2e}
\usepackage{algorithmicx}

\DontPrintSemicolon

\usepackage{amsmath}
\usepackage{amssymb}
\usepackage{amsthm}

\usepackage{thmtools}
\usepackage{thm-restate}

\usepackage{xspace}
\usepackage{authblk}
\usepackage[color=green!20]{todonotes}

\newtheorem{theorem}{Theorem}

\newtheorem{lemma}[theorem]{Lemma}
\newtheorem{claim}[theorem]{Claim}
\newtheorem*{thm*}{Theorem}

\theoremstyle{remark}
\newtheorem*{remark}{Remark}
\newtheorem*{eth-env}{Exponential Time Hypothesis}

\def\dist{\ensuremath{\text{dist}}\xspace}
\def\start{\ensuremath{\text{start}}\xspace}
\def\path{\ensuremath{\text{path}}\xspace}
\def\circle{\ensuremath{\text{circle}}\xspace}
\def\L{\ensuremath{\mathcal{L}}\xspace}
\def\K{\ensuremath{\mathcal{K}}\xspace}
\def\R{\ensuremath{\mathbb{R}}\xspace}
\def\rec{\ensuremath{\textbf{Recognition}}\xspace}
\def\stretch{\ensuremath{\textbf{Stretchability}}\xspace}

\DeclareMathOperator{\NP}{NP}
\DeclareMathOperator{\PSPACE}{PSPACE}

\newcommand{\gcirc}{\ensuremath{\mathcal{C}}}
\newcommand{\A}{\ensuremath{\mathcal{A}}}
\newcommand{\gline}{\ensuremath{\ell}}

\title{Intersection Graphs of Rays and Grounded Segments}
 
\author[1]{Jean Cardinal}
\author[2]{Stefan Felsner}
\author[3]{Tillmann Miltzow\footnote{supported by the ERC grant PARAMTIGHT: "Parameterized complexity and the search for tight complexity results", no. 280152.}}
\author[4]{Casey~Tompkins}
\author[5]{Birgit Vogtenhuber}
\affil[1]{Universit\'e libre de Bruxelles (ULB), 
			Brussels, Belgium \ \texttt{jcardin@ulb.ac.be}}
\affil[2]{Institut f\"ur Mathematik, Technische Universit\"at Berlin (TU), 
			Berlin, Germany \ \texttt{felsner@math.tu-berlin.de}} 
\affil[3]{Institute for Computer Science and Control,
 			Hungarian Academy of Sciences (MTA SZTAKI), Budapest, Hungary \ \texttt{t.miltzow@gmail.com}}
\affil[4]{Alfr\'ed R\'enyi Institute of Mathematics, 
			Hungarian Academy of Sciences, Budapest, Hungary \ \texttt{ctompkins496@gmail.com}}
\affil[5]{Institute of Software Technology, 
			Graz University of Technology, Graz, Austria \ \texttt{bvogt@ist.tugraz.at}}
 
\date{}

\begin{document}

\maketitle

\begin{abstract}
  We consider several classes of intersection graphs of line segments in the plane and prove new equality and separation results between those classes. 
  In particular, we show that:
  \begin{itemize}
  \item intersection graphs of grounded segments and intersection graphs of downward rays form the same graph class,
  \item not every intersection graph of rays is an intersection graph of downward rays, and
  \item not every intersection graph of rays is an outer segment graph.
  \end{itemize}
  The first result answers an open problem posed by Cabello and Jej\v{c}i\v{c}. The third result confirms a conjecture by Cabello.
  We thereby completely elucidate the remaining open questions on the containment relations between these classes of segment graphs.
  We further characterize the complexity of the recognition problems for the classes of outer segment, grounded segment, and ray intersection graphs. 
  We prove that these recognition problems are complete for the existential theory of the reals. 
  This holds even if a 1-string realization is given as additional input.
\end{abstract}
\sloppy


\section{Introduction}

Intersection graphs encode the intersection relation between objects in a collection.
More precisely, given a collection $\A$ of sets, the induced intersection graph has the collection $\A$ as the set of vertices, and two vertices $A,B\in \A$ are adjacent whenever $A\cap B\not= \emptyset$.
Intersection graphs have drawn considerable attention in the past thirty years, to the point of constituting a whole subfield of graph theory (see, for instance, the book from McKee and Morris \cite{MM99}).
The roots of this subfield can be traced back to the properties of interval graphs -- intersection graphs of intervals on a line -- and their role in the discovery of the linear structure of bacterial genes by Benzer in 1959~\cite{benzer59}.

We consider {\em geometric} intersection graphs, that is, intersection graphs of simple geometric objects in the plane, such as curves, disks, or segments.
While early investigations of such graphs are a half-century old~\cite{S66}, the modern theory of geometric intersection graphs was established in the nineties by Kratochv{\'{\i}}l~\cite{K91,K91a}, and Kratochv{\'{\i}}l and Matou\v{s}ek~\cite{KM91,KM94}.
They introduced several classes of intersection graphs that are the topic of this paper.
Geometric intersection graphs are now ubiquitous in discrete and computational geometry, and deep connections to other fields such as complexity theory~\cite{SSS03,S09,M14} and order dimension theory~\cite{CHOSU14,F14,CFHW15} have been established.

We will focus on the following classes of intersection graphs, most of which are subclasses of intersection graphs of line segments in the plane, or {\em segment (intersection) graphs}.

\paragraph*{Grounded Segment Graphs}
Given a \emph{grounding line} $\gline$, we call a segment $s$ a \emph{grounded segment} if one of its endpoints, called the base point, is on $\gline$ and the interior of $s$ is above $\gline$. 
A graph $G$ is a \emph{grounded segment graph} if it is the intersection graph of a collection of grounded segments (w.r.t.\ the same grounding line $\gline$).

\paragraph*{Outer Segment Graphs}
Given a \emph{grounding circle} $\gcirc$, a segment $s$ is called an \emph{outer segment} if exactly one of its endpoints, called the base point, is on $\gcirc$ and the interior of $s$ is inside~$\gcirc$. 
A graph $G$ is an \emph{outer segment graph} if it is the intersection graph of a collection of outer segments (w.r.t.\ the same grounding circle $\gcirc$).

\paragraph*{Ray Graphs and Downward Ray Graphs}
A graph $G$ is a \emph{ray graph} if it is the intersection graph of rays (halflines) in the plane.
A ray $r$ is called a \emph{downward ray} if its apex is above all other points of $r$.
A graph $G$ is a \emph{downward ray graph} if it is the intersection graph of a collection of downward rays.
It is not difficult to see that every ray graph is also an outer segment graph: consider a grounding circle at infinity.
Similarly, one can check that downward ray graphs are grounded segment graphs.

\paragraph*{String Graphs}
{\em String graphs} are defined as intersection graphs of collections of simple curves in the plane with no
three intersecting in the same point. 
We refer to these curves as strings.
In this treatment we consider only {\em 1-string graphs}, this is, string graphs such that two strings pairwise intersect at most once.
We define {\em outer $1$-string graphs} and {\em grounded $1$-string graphs} 
in the same way as for segments. 

\medskip
If clear from context, we refer to the class of ray graphs, just as rays; the class of grounded segment graphs just as grounded segments and so on.

\paragraph*{Containment and Separation}
We want to point out that all graph classes considered here are intrinsically similar. Their common features are a representation of $1$-dimensional objects in the plane ``attached'' to a common object. In case of downward ray graphs and ray graphs this common object is the ``line at infinity'' or the ``circle at infinity''  respectively.
This intrinsic similarity is reflected nicely by proofs, which are unified, elegant and simple. 

In a recent manuscript, Cabello and Jej\v{c}i\v{c} initiated a comprehensive study aiming at refining our understanding of the containment relations between classes of geometric intersection graphs involving segments, disks, and strings~\cite{CJ16,CJ16abs}.
They introduce and solve many open questions about the containment relations between various classes.
In particular, they prove proper containment between intersection graphs of segments with $k$ or $k+1$ distinct lengths, intersection graphs of disks with $k$ or $k+1$ distinct radii, and intersection graphs of outer strings and outer segments.
In their conclusion~\cite{CJ16}, they leave open two natural questions:
\begin{itemize}
\item Is the class of outer segment graphs a proper subclass of ray graphs?
\item Is the class of downward ray graphs a proper subclass of grounded segment graphs?
\end{itemize}

In this contribution, we answer the first question in the positive, thereby proving a conjecture of Cabello.
We also give a negative answer to the second question by showing that downward rays and grounded segments yield the same class of intersection graphs.
We henceforth completely settle the remaining open questions on the containment relations between these classes of segment graphs.
We summarize the complete containment relationship in the following theorem, see also Figure~\ref{fig:summary} for an illustration.
Note that Item~\ref{itm:Sergio} was proved already by Cabello and Jej\v{c}i\v{c} and Item~\ref{itm:StringsEqual} can be seen as folklore. 

\begin{theorem} \label{thm:Containment} The following containment relations of intersection graphs hold:
  \begin{enumerate}
  \item grounded segments  $=$ downward rays,
  \item downward rays  $\subsetneq$ rays,
  \item rays  $\subsetneq$ outer segments, 
  \item \label{itm:Sergio} outer segments  $\subsetneq$ outer $1$-strings and 
  \item \label{itm:StringsEqual}  outer $1$-strings  $=$ grounded $1$-strings.
  \end{enumerate}
\end{theorem}
\begin{figure}[htbp]
  \centering
  \includegraphics{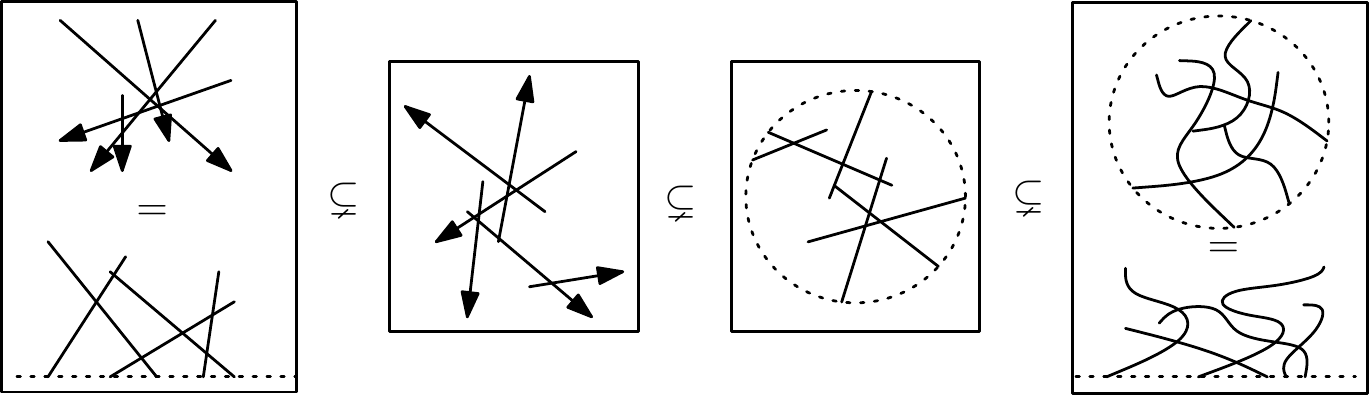}
  \caption{\label{fig:summary}Schematic description of some of our results.}
\end{figure}


\paragraph*{The Complexity Class $\exists\R$ and the Stretchability Problem}

The complexity class $\exists\R$ is the collection of decision problems that are polynomial-time equivalent to deciding the
truth of sentences in the first-order theory of the reals of the form:
$$
\exists x_1 \exists x_2 \ldots \exists x_n F(x_1,x_2,\ldots ,x_n),
$$
where $F$ is a quantifier-free formula involving inequalities and equalities of polynomials in the real variables $x_i$.
This complexity class can be understood as a ``real'' analogue of $\NP$.
It can easily be seen to contain $\NP$, and is known to be contained in $\PSPACE$~\cite{C88}.

In recent years, this complexity class revealed itself most useful for characterizing the complexity of realizability problems in computational geometry.
A standard example is the {\em pseudoline stretchability problem}.

Matou{\v{s}}ek~\cite[page 132]{matouvsek2002lectures} defines an \emph{
arrangement of pseudolines} as a finite collection of curves in the plane that satisfy the following conditions:
\begin{enumerate}[(i)]
\item  Each curve is $x$-monotone and unbounded in both directions.
\item  Every two of the curves intersect in exactly one point, and they cross at the intersection.
\end{enumerate}

In the stretchability problem, one is given the combinatorial structure of an arrangement of pseudolines in the plane as input, and is asked whether the same combinatorial structure can be realized by an arrangement of {\em straight lines}.
If this is the case, then we say that the arrangement is {\em stretchable}.
This structure can for instance be given in the form of a set of $n$ {\em local sequences}: the left-to-right order of the intersections of each line with the $n-1$ others. 
Equivalently, the input is the underlying rank-3 oriented matroid.
The stretchability problem is known to be $\exists\R$-complete~\cite{S90}.
We refer the reader to the surveys by Schaefer~\cite{S09}, Matou{\v{s}}ek~\cite{M14}, and Cardinal~\cite{C15} for further details.

\paragraph*{Computational Complexity Questions}
Given a graph class $\mathcal{G}$, we define $\rec(\mathcal G)$ as the following decision problem:
\begin{framed}
  \noindent \textbf{\rec$(\mathcal{G})$}
  
  \vspace{0.2cm}
  
  \noindent \textbf{Input:} A graph $G = (V,E)$. \\
  \noindent \textbf{Question:} Does $G$ belong to the graph class $\mathcal G$? 
\end{framed}

Potentially the recognition problem could become easier if we have some additional information.
In our case it is natural to ask if a given outer $1$-string representation of a graph $G$ has an outer segment representation. 
The same goes for grounded $1$-strings and grounded segments.
Finally, we will consider outer $1$-strings and rays.
Formally, we define the decision problem $\stretch(\mathcal{G},\mathcal{F})$ as follows.

\begin{framed}
\noindent \textbf{\stretch$(\mathcal{G},\mathcal{F})$}
  
  \vspace{0.2cm}
  
  \noindent \textbf{Input:} A graph $G = (V,E)$ and representation $R$ that shows that $G$ belongs to $\mathcal{F}$. \\
  \noindent \textbf{Question:} Does $G$ belong to the graph class $\mathcal G$? 
\end{framed}
Note that we need to assume that $\mathcal{F}$ is a graph class defined by intersections of certain objects.

We also complete the picture by giving computational hardness results on recognition and stretchability questions by proving the following theorem.
\begin{restatable}{theorem}{Complexity}\label{thm:RecHardness}
    The following problems are $\exists\R$-complete:
    \begin{itemize}
   \item $\rec(\textrm{outer segments})$ 
   and $\stretch(\textrm{outer segments}, \textrm{ outer $1$-strings})$,
   \item $\rec(\textrm{grounded segments})$
   and \\ $\stretch(\textrm{grounded segments},\textrm{ grounded $1$-strings})$,
   \item $\rec(\textrm{rays})$
   and $\stretch(\textrm{rays}, \textrm{outer $1$-strings})$.
  \end{itemize}
\end{restatable}
\medskip
We want to point out that all statements of Theorem~\ref{thm:RecHardness} are proven in one simple and unified way. This uses heavily the complete chain of containment of the graph classes and the intrinsic similarity of all considered graph classes.
A highlight of Theorem~\ref{thm:RecHardness} is certainly the $\exists\R$-complete on the recognition problem of the natural graph class of ray intersection graphs.
Note that this strengthens the result of Cabello and Jej\v{c}i\v{c} on the separation between outer $1$-string and outer segment graphs.

\paragraph*{Previous Work and Motivation}

The understanding of the inclusion properties and the complexity of the recognition problem for classes of geometric intersection graphs have been the topic of numerous previous works.

Early investigations of string graphs date back to Sinden~\cite{S66}, and Ehrlich, Even, and Tarjan~\cite{EET76}. 
Kratochv{\'{\i}}l~\cite{K91} initiated a systematic study of string graphs, including the complexity-theoretic aspects~\cite{K91a}.
It is only relatively recently, however, that the recognition problem for string graphs has been identified as NP-complete~\cite{SSS03}.
NP membership is far from obvious, given that there exist string graphs requiring exponential-size representations~\cite{KM91}.

Intersection graphs of line segments were studied extensively by Kratochv{\'{\i}}l and Matou\v{s}ek~\cite{KM94}.
In particular, they proved that the recognition of such graphs was complete for the existential theory of the reals.
A key construction used in their proof is the {\em Order-forcing Lemma}, which permits the embedding of pseudoline arrangements as segment representations of graphs.
Some of our constructions can be seen as extensions of the Order-forcing Lemma to grounded and outer segment representations.

Outer segment graphs form a natural subclass of outer string graphs as defined by Kratochv{\'{\i}}l~\cite{K91}. 
They also naturally generalize the class of {\em circle graphs}, which are intersection graphs of chords of a circle~\cite{N85}.

A recent milestone in the field of segment intersection graphs is the proof of Scheinerman's conjecture by Chalopin and Gonçalves~\cite{CG09}, stating that planar graphs form a subclass of segment graphs. 
It is also known that outerplanar graphs form a proper subclass of circle graphs~\cite{WP85}, hence of outer segment graphs. 
Cabello and Jej\v{c}i\v{c}~\cite{CJ16} proved that a graph is outerplanar if and only if its 1-subdivision is an outer segment graph.

Intersection graphs of rays in two directions have been studied by Soto and
Telha~\cite{ST11}, they show connections with the jump number of some posets
and hitting sets of rectangles. The class has been further studied by Shrestha
et al.~\cite{STU10}, and Musta\c{t}\u{a} et al.~\cite{MNTTU16}. The results
include polynomial-time recognition and isomorphism algorithms. This is in
contrast with our hardness result for arbitrary ray graphs.

Properties of the chromatic number of geometric intersection graphs have been
studied as well.  For instance, Rok and Walczak proved that outer string
graphs are $\chi$-bounded~\cite{RW14}, and Kostochka and
Ne\v{s}et\v{r}il~\cite{KN98,KN02} studied the chromatic number of ray graphs
in terms of the girth and the clique number.

The complexity of the maximum clique and independent set problems on classes
of segment intersection graphs is also a central topic of study.
It has been shown recently, for instance, that the maximum clique problem is NP-hard on ray graphs~\cite{CCL13}, and that the maximum independent set problem is polynomial-time tractable on outer segment graphs~\cite{KMPV17}.

\paragraph*{Organization of the Paper}

In the next section, we give some basic definitions and observations. We also provide a short proof of the equality between the classes of
downward ray and of grounded segment graphs.

In Section \ref{sec:cycle}, we introduce the {\em Cycle Lemma}, a construction that will allow us to control the order of the slopes
of the rays in a representation of a ray graph, and the  order in which the segments are attached to the grounding line
or circle in representations of grounded segment and outer segment graphs.

In Section \ref{sec:stretch}, we show how to use the Cycle Lemma to encode the
pseudoline stretchability problem in the recognition problem for outer
segment, grounded segment, and ray graphs. We thereby prove that those
problems are complete for the existential theory of the reals and thus Theorem~\ref{thm:RecHardness}.

Finally, in Section \ref{sec:raySeg}, we establish two new separation results. 
First, we prove that ray graphs form a proper subclass of outer segment graphs, proving Cabello's conjecture.
Then we prove that downward ray graphs form a proper subclass of ray graphs.
We want to point out that Lemma~\ref{lem:DRayEqualGSegment}, Lemma~\ref{lem:OuterStringEqualGroundedString}, Theorem~\ref{thm:RayNotOuter} and Theorem~\ref{thm:RayNotDownwardRay} together show Theorem~\ref{thm:Containment}

\section{Preliminaries}\label{sec:preliminaries}

We first give a short proof of the equality between the classes of ray and
grounded segment graphs, thereby answering Cabello and Jej\v{c}i\v{c}'s second
question.  The proof is illustrated in Figure~\ref{fig:groundedtoray}.

 \begin{figure}[htbp]
\centering
\includegraphics{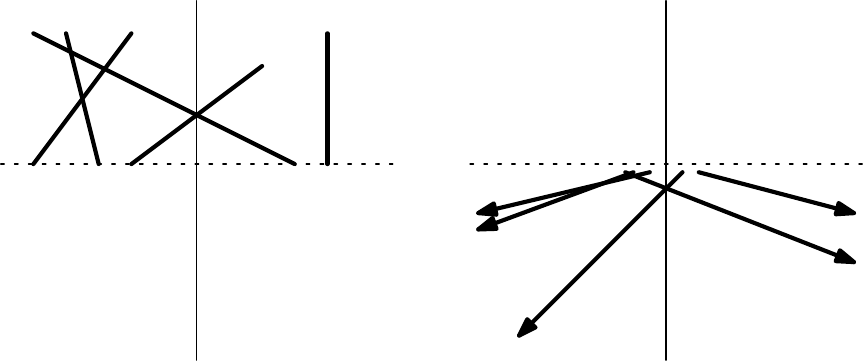}
\caption{Transforming a collection of grounded segments (left) into downward rays (right).
}\label{fig:groundedtoray}
\end{figure}

%

\begin{lemma}[Downward Rays $=$ Grounded Segments]\label{lem:DRayEqualGSegment}
  A graph $G$ can be represented as a grounded segment graph if and only if it can be represented by downward rays.
\end{lemma}
\begin{proof}
Consider a coordinate system where the grounding line is the $x$-axis, and take the projective transformation defined in homogeneous coordinates by
  \[\begin{pmatrix}
     x\\ y\\ 1
    \end{pmatrix}
    \mapsto 
    \begin{pmatrix}
     x\\ -1\\ y
    \end{pmatrix}.
    \]
This projective transformation is a bijective mapping from the projective plane to itself, which maps grounded segments to downward rays. 
In the plane, it can be seen as mapping the points $(x,y)$ with $y>0$ to $(x/y, -1/y)$. 
Since projective transformations preserve the incidence structure, the equivalence of the graph classes follows. \qedhere
\end{proof}

\begin{lemma}[Ray Characterization]\label{lem:Raycharacterization}
  A graph $G$ can be represented as an outer segment graph with all intersections of line extensions inside the grounding circle $\gcirc$ if and only if it can be represented by rays.
\end{lemma}
\begin{proof}
  See Figure~\ref{fig:RayEqualOuter} for an illustration of the following.

  ($\Leftarrow$) Let $R$ be a representation of $G$ by rays, and let \L be the set of the lines extending all involved rays. Then there exists a circle $\gcirc$ that contains all the intersections of \L and at least some part of each ray. We define a representation $R'$ of $G$ as outer segment representation by restricting each ray to the inside of $\gcirc$. It is easy to see that this  indeed is a representation of $G$ with the desired property.
  
  ($\Rightarrow$) Let $R$ be a representation of $G$ by outer segments. We define a set of rays by just extending each segment at its base point on the grounding circle $\gcirc$ to a ray. If two segments intersected before, then the corresponding rays will intersect as well trivially.  Moreover, by the assumption that all the line extensions intersect inside $\gcirc$, it follows that the rays will not intersect outside $\gcirc$s, and hence the cooresponding ray graph is a representation of $G$.
\end{proof}

\begin{figure}[htbp]
\centering
\begin{minipage}[b]{.4\linewidth}
\centering
\includegraphics{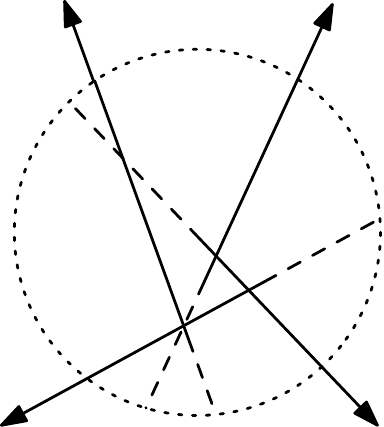}
\subcaption{Given a collection of rays, there exists a circle that contains all intersections of all extensions of the rays and thus induces an outer segment representation of the same graph.}\label{fig:RayEqualOuter}
\end{minipage}%
\hspace{0.3cm}
\begin{minipage}[b]{.55\linewidth}
\centering
\includegraphics{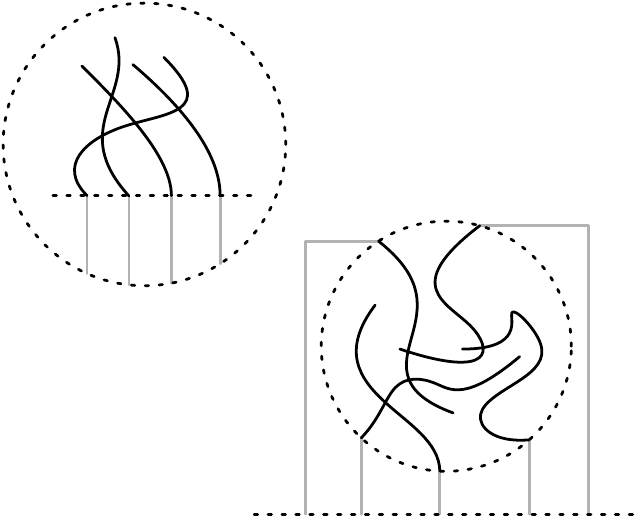}
\subcaption{Simple extensions of outer $1$-strings yield grounded $1$-strings and vice versa.}\label{fig:StringEqualString}
\vspace{0.7cm}
\end{minipage}

\caption{Simple transformations between different graph representations.}\label{fig:1}
\end{figure}

\begin{remark}
  It is tempting to try to find a projective transformation that maps the unit circle $S^1$ to infinity in a way that outer segments become rays. 
As we will show later, outer segments and rays represent different graph classes. Thus such a mapping is impossible. 
With the help of M\"{o}bius transformations it is possible to find a mapping that maps the unit circle $S^1$ to infinity. 
However, outer segments then become connected parts of hyperbolas instead of straightline rays. 
\end{remark}

Recall that a collection of curves in the plane are called $1$-strings if any two curves intersect \emph{at most} once.
We can define \emph{grounded $1$-string graphs} and \emph{outer $1$-string graphs} in an analogous way to the corresponding segment graphs by replacing segments by $1$-strings.

\begin{lemma}[Grounded 1-Strings $=$ Outer 1-Strings]\label{lem:OuterStringEqualGroundedString}
  A graph $G$ can be represented as a grounded $1$-string graph if and only if it can be represented as an outer $1$-string graph.
\end{lemma}

\begin{proof}
    See Figure~\ref{fig:StringEqualString} for an illustration of this proof.
    
    ($\Rightarrow$) Let $R$ be a representation of $G$ by grounded $1$-strings with grounding line $\gline$.  
	Take a large circle $\gcirc$ that completely contains 
	$R$ and extend the $1$-strings perpendicularly from the grounding point on $\gline$ to the opposite side of $\gline$ until they meet the circle $\gcirc$. 
	This procedure yields an outer $1$-string representation with grounding circle $\gcirc$ and the same incidences as $R$, hence an outer $1$-string representation of $G$.
    
    ($\Leftarrow$) Let $R$ be a representation of $G$ by outer $1$-strings grounded on a circle $\gcirc$. 
	Let $\gline$ be a horizontal line below $\gcirc$. 
	Extend any $1$-string whose grounding point is on the bottom half of $\gcirc$ with a vertical line segment to $\gline$.  
	Extend any $1$-string whose grounding point is on the top half of $\gcirc$ with a horizontal segment followed by a vertical segment from $\gcirc$ to the line $\gline$.  
	This procedure clearly does not alter any incidences. 
	Thus it provides a grounded segment representation of $G$.
\end{proof}

\paragraph*{Ordered Representations}
Given a graph $G$ and a permutation $\pi$ of the vertices, we say that a grounded (segment or string) representation of $G$ is \emph{$\pi$-ordered} if the base points of the cooresponding segments or strings are in the order of $\pi$ on the grounding line, up to inversion and cyclic shifts.
In the same fashion, we define \emph{$\pi$-ordered} for outer (segment or string) representations and (downward) ray representations, where rays are ordered by their angles with the horizontal axis.


\section{Cycle Lemma} \label{sec:cycle}

For some of our constructions, we would like to force that the segments or strings representing the vertices of a graph appear in a specified order on the grounding line or circle.
To this end, we first study some properties of the representation of cycles, which in turn will help us to enforce this order.

Given a graph $G = (V,E)$ on $n$ vertices $V = \{v_1,\dots,v_n\}$ and a permutation $\pi$ of the vertices of $G$, we define the \emph{order forcing graph} $G^{\pi}$ as follows.
The vertices $V(G^{\pi})$ are defined by $V \cup \{1,\ldots,2 n^2 \}$
and the edges $E(G^\pi)$ are defined by 
\( E \cup \{\, (2in, v_{\pi(i)} )\, | \, i= 1,\ldots , n \, \} \cup 
\{\, (i, i+1 )\, | \, i= 1,\ldots , 2n^2 \, \}\)
(here we conveniently assume $2n^2+1 = 1$).

  For the sake of simplicity, we think of $\pi$ as being the identity and
  the vertices as being indexed in the correct way.
  The vertices of $G$ are called \emph{relevant}, and the additional vertices of $G^{\pi}$ are called {\em cycle vertices}.
  Note that on the cycle, the distance between any two cycle vertices $u,v$ that are adjacent to different relevant vertices is at least $2n$.
  
\begin{figure}[htbp]
  \centering
  \includegraphics{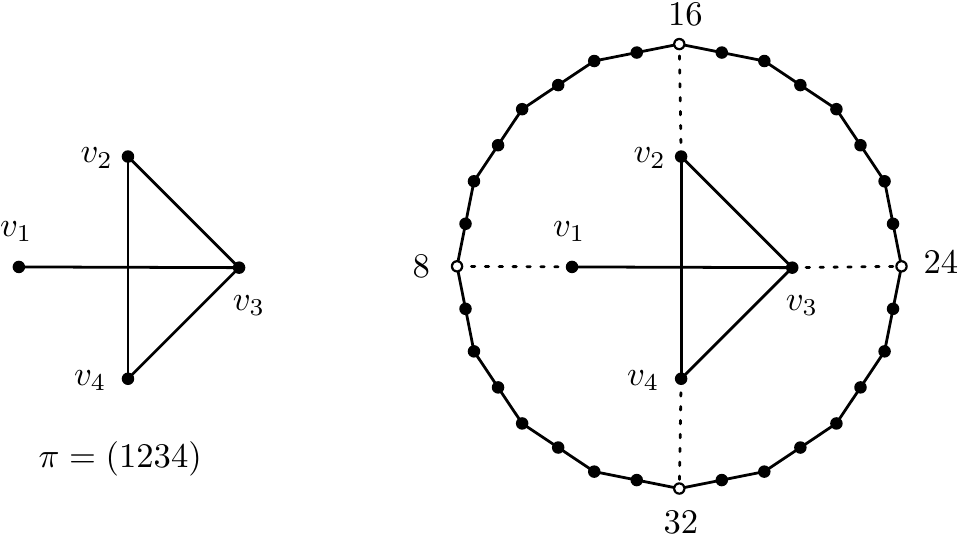}
  \caption{Illustration of the definition of order forcing graphs.}
  \label{fig:OrderForcingGraph}
\end{figure}
  
\begin{lemma}[Cycle Lemma]\label{lem:cycle}
 Let $G$ be a graph and $\pi$ be a permutation of the vertices of $G$.
 Then there exists a $\pi$-ordered representation of $G$ if and only if there exists a representation of $G^{\pi}$.
 This is true for the following graph classes: grounded segment graphs, ray graphs, outer segment graphs, and outer $1$-string graphs.
\end{lemma}
Note that for the case that $|V(G)| \leq 3$ this statement is trivial, as it can be easily checked that in these finitely many cases both graphs can always be realized.
Thus from now on, we assume that $|V(G)| \geq 4$.

Before proving Lemma~\ref{lem:cycle}, 
we first study the representations of cycles.
Let $C = 1,2,3,\ldots ,n$ be a cycle of length $n$ and $R$ be a $1$-string representation of $C$.
Then each string $i$ is crossed by the strings $(i-1)$ and $(i+1)$ exactly once.
The part of $i$ between the two intersections is called \emph{central part} of $i$ and denoted by $z_i$.
The intersection points are denoted by $p_{i,i-1}$ and $p_{i,i+1}$.

\begin{figure}[htbp]
  \centering
  \includegraphics{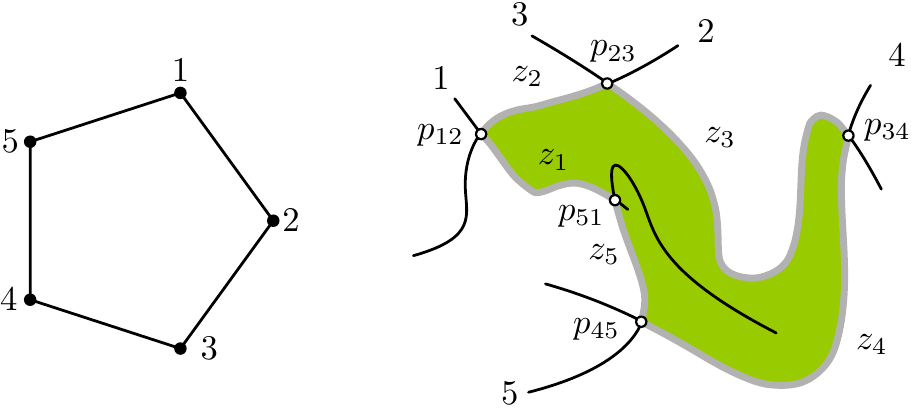}
  \caption{A 1-string representation of a 5-cycle. 
	}
  \label{fig:cycle}
\end{figure}

\begin{lemma}\label{lem:centralPartOfCycle}
    Let $C$ be a cycle and $R$ be a $1$-string representation of $C$.
    The union of all central parts of all the $1$-strings of $R$ forms a Jordan curve, which we denote by $J(C)$.
    This also holds in case that $C$ is a subgraph of some other graph $G$.
\end{lemma}
\begin{proof}
    Using the above notation, the curve can be explicitly given as: \[p_{12},z_2 ,p_{23},z_3,\ldots ,z_n, p_{n1},z_1. \qedhere\]
\end{proof}


\begin{lemma}\label{lem:InnerIntersection}
    Let $C$ be an induced cycle of the graph $G$ and $R$ be an outer $1$-string representation of $G$. 
    Further, let $a,b \notin C$ be two adjacent vertices, which are adjacent to two vertices $u_a,u_b \in V(C)$ with $\dist(u_a,u_b) \geq 4$ on the cycle.
	Then $a$ must intersect the central part of $u_a$, $b$ must intersect the central part of $u_b$, and $a$ and $b$ must intersect in the interior of $J(C)$.
\end{lemma}

\begin{figure}[htbp]
  \centering
  \includegraphics{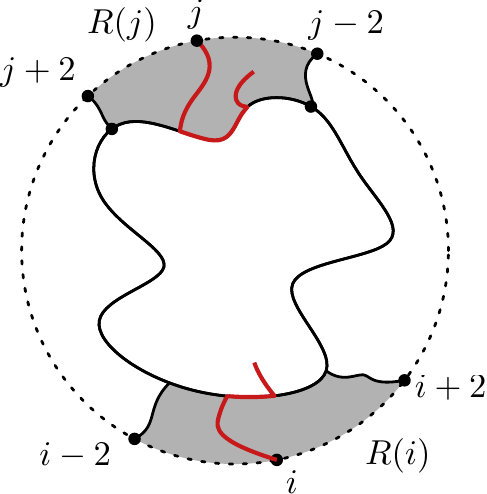}
  \caption{Illustration of Lemma~\ref{lem:InnerIntersection}.}
  \label{fig:CycleRegions}
\end{figure}

\begin{proof}
  Let $i\in V(C)$ be some outer $1$-string.
  We define by $\start(i)$ the portion of $i$ between its base point and the first point on the central part, which we denote by $p_i$. 
  Given three distinct points $p,q,r$ on $J(C)$, we denote by $\path(p,q,r)$ the portion of $J(C)$ bounded by $p$ and $q$ and containing $r$.
  Similarly, let $p,q,r$ be three distinct points on the grounding circle.
  Then there exists a unique portion $\circle(p,q,r)$ of the grounding circle bounded by $p$ and $q$ and containing $r$.
  For each $i\in V(C)$, we consider the region $R(i)$ bounded by the following four curves: 
    \[
     \start(i-2)\ ,\ 
     \path(p_{i-2},p_{i+2},p_i) \ ,\ 
     \start(i+2) \ ,\ 
     \circle(i-2,i+2,i)
    \]
    We summarize a few useful facts on these regions.
    \begin{enumerate}
     \item String $i$ is contained in the union of the region $R(i)$ and the interior of $J(C)$.
     \item If $\dist(i,j) \geq 4$ then $R(i)$ and $R(j)$ are interior disjoint and $i \cap R_j = \emptyset$ .
     \item If $v\notin V(C)$ is adjacent to $i\in V(C)$ but not adjacent to any other $j\in V(C)$, then the base point of $v$ must be inside $R(i)$.
    \end{enumerate}
    The first statement follows from the fact that $i$ is disjoint from $i-2$ and $i+2$.    
    The second statement follows immediately from the definition of $R(i)$.
    The last statement follows from the fact that there is no way to reach the central part or $R(i)$ in case that $v$ does not start in $R(i)$.
    From here follows the proof. 
    The outer $1$-string $a$ must have its base point in $R(u_a)$ and $b$ must have its base point in $R(u_b)$.
    These two regions are interior disjoint and don't have a common boundary formed by any part of $u_a$ or $u_b$.
    Hence, as $a$ and $b$ are not allowed to cross any other string of $C$, they can only intersect in the interior of $J(C)$.
\end{proof}
%

%
We are now ready to prove our main lemma.

\begin{figure}[htbp]
  \centering
  \includegraphics{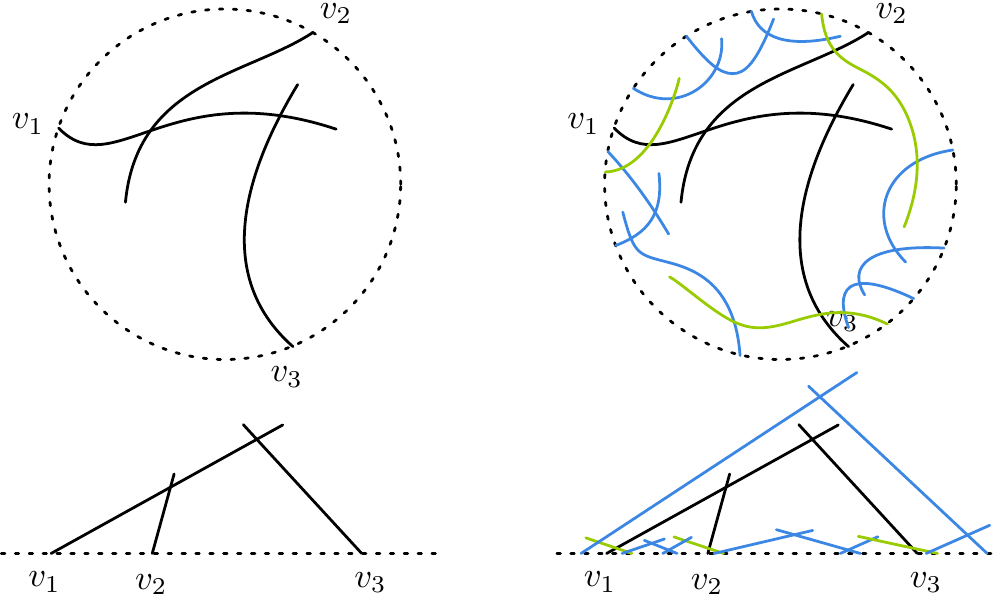}
  \caption{
		Representations of order forcing graphs for ordered representations of outer string graphs and grounded segment graphs.}
  \label{fig:AddingCycles}
\end{figure}

\begin{proof}[Proof of Lemma~\ref{lem:cycle}]
  $(\Rightarrow)$ Let $R$ be an ordered representation of $G$ with respect to $\pi$. We have to construct a representation $R^\pi$ of the graph $G^\pi$. 
  
  The two cases of outer $1$-string graphs and grounded segment graphs follow the same simple pattern; 
  see \Cref{fig:AddingCycles} for a pictoral proof of those cases. 
  Here, we directly proceed to showing the case of outer segment graphs, as it is more involved.
  Let $R$ be an ordered outer segment representation of $G$ with respect to $\pi$ with grounding circle $\gcirc$. 
  Modify $R$ such that each segment stops at its last intersection point.
  In order to construct $R^\pi$, we will add the segments representing the cycle vertices to $R$. 
  As a first step, we add to each segment $s$ of $R$ a tiny segment $t_s$ orthogonal to $s$, intersecting it and being very close to it. These tiny segments are small enough so that they intersect no other segment, but are on both sides of $s$.
  Let $u,v \in V(G)$ 
  be two segments that are successive in the order $\pi$, that is, the base points $u$ and $v$ are consecutive on the grounding circle $\gcirc$.
%
  
%
  We denote by $E$ the region inside $\gcirc$ and outside the convex hull of all the segments; 
  see the illustration in the middle of \Cref{fig:ForcedOrderOuterSegment}. 
  Note that $E$ has at most $(n-2)/2$ reflex vertices $r_1,\ldots,r_k$.
  We extend $k$ rays from the center of $\gcirc$ through $r_1,\ldots,r_k$. This divides $E$ in at most $k+2 < n-2$ regions.
  Recall that we assume $n\geq 4$.
  It is easy to see that we can place one grounded segment into each region such that they form the desired path from $u$ to $v$ without intersecting any other segment from $R$.
  In order to obtain enough segments on the path, it might be necessary to use several segments inside one region. 
  
\begin{figure}[htbp]
  \centering
  \includegraphics{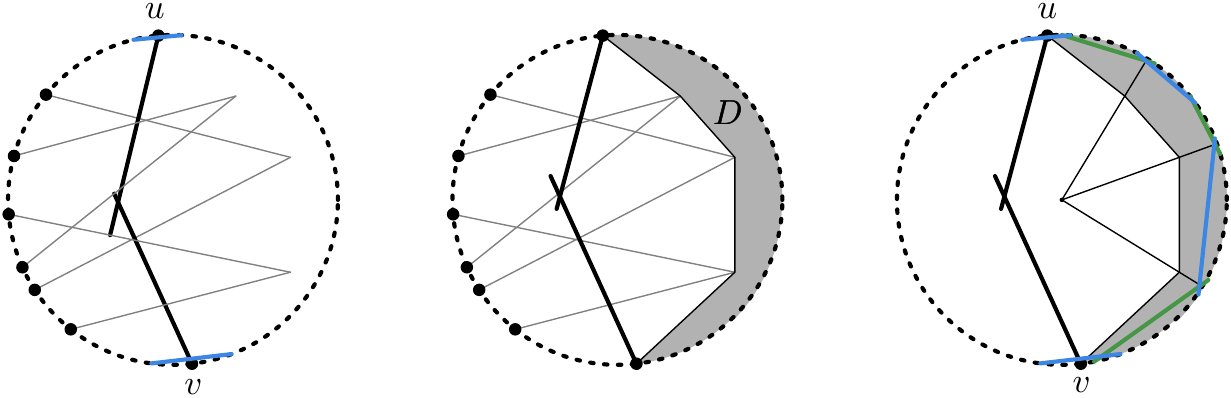}
  \caption{Representation of $G^\pi$ as outer segment graph.}
  \label{fig:ForcedOrderOuterSegment}
\end{figure}

  Now we show the statement for ray graphs; cf.~\Cref{fig:AddingCycleRays}. We start with a representation $R$
  of our ordered ray graph of $G$ with respect to $\pi$.
  Let $D$ be a sufficiently large disk that contains all ray starting points
  as well as all intersections among all the rays and let $\partial D$ be the
  boundary of $D$.
  For each ray $r$ we define $\ell_r$ to be the line orthogonal to
  $r$ through the unique point $\partial D \cap r $.
  Note that $\ell_r$ is usually not tangent to $\partial D$.
  We choose $D$ large enough such that the collection of all lines $\ell_r$ defined in this way determine a \emph{convex} polygon $P$ in which each $\ell_r$ is the supporting line of an edge of $P$. 
  The \emph{convex} polygon $Q$ is defined by adding $n$ sufficiently small edges in the vicinity of each vertex of $P$.
  (Recall that the distance between two consecutive vertices on the cycle is $2n$.)
  Now denote with $v_1,\ldots,v_{k}$ the $n + n^2$ orthogonal vectors of the edges $e_1,\ldots,e_{k}$ of $Q$ in clockwise order. 
  We place at each edge $e_i$ two rays $q_i$ and $r_i$ with slopes $v_{i+1}$ and $v_i$ so that $q_i$ and $r_i$ intersect (their apices being close to the endpoints of the edges and not on any ray of $R$).  
  This is illustrated with a regular $k$-gon at the bottom right of \Cref{fig:AddingCycleRays}. 
  It is easy to see that the intersection graph of these rays is a cycle, and that each ray of $R$ intersects exactly one of the new rays.
  The representation $R^{\pi}$ of $G^\pi$ is the union of the rays of $R$ and the newly defined rays. 
  
  ($\Leftarrow$) 
  Recall that we have to show the following. If $G^\pi$ has representation, then $G$ also has a $\pi$-ordered representation. We show this by considering a representation $R^\pi$ of $G^\pi$. 
  It is clear that $R^\pi$ restricted to the relevant vertices gives a representation of $G$.
  We will show that the vertices $V(G)$ are $\pi$-ordered.
  
  We will restrict considerations to outer $1$-string graphs for notational convenience, as the proof is essentially the same for each graph class.
  (Recall that all of these graph classes are contained in the class of outer $1$-string graphs. Further, we can produce an outer $1$-string representation from a ray, grounded segment or outer segment representation.)
%
  By Lemma~\ref{lem:InnerIntersection}, each relevant outer $1$-string adjacent to the circle vertex $i$ is fully contained in the region $R(i)$, as described in the proof of Lemma~\ref{lem:InnerIntersection}. As all regions $R(2n),R(4n),R(6n), \ldots $ are pairwise disjoint and arranged in this order on the grounding circle, this order is also enforced on the $1$-strings of $V(G)$. 
\end{proof}

\begin{figure}[htbp]
  \centering
  \includegraphics{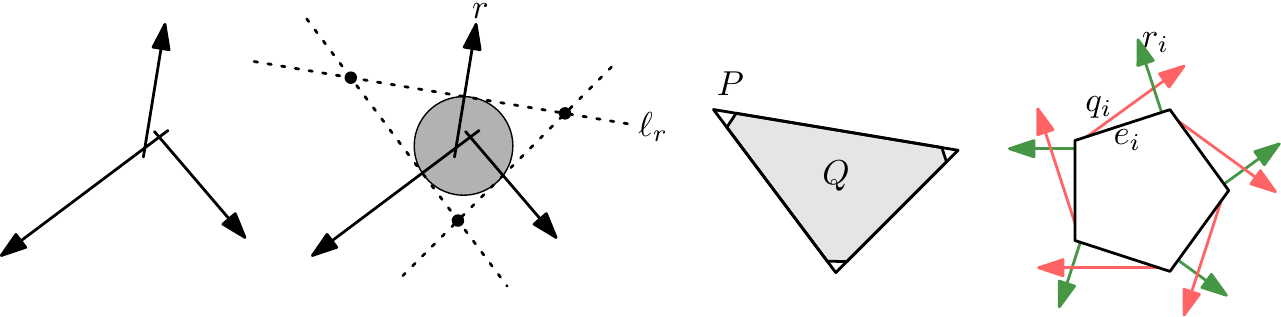}
  \caption{Illustration of the proof of Lemma~\ref{lem:cycle} for rays.
}
  \label{fig:AddingCycleRays}
\end{figure}

\section{Stretchability}\label{sec:stretch}

The main purpose of this section is to show that the recognition
of the graph classes defined above is $\exists\R$-complete. 
For this we will use Lemma~\ref{lem:cycle} extensively.
It is likely that our techniques can be applied to other graph classes as well.

\Complexity*

\begin{proof}
  By Lemma~\ref{lem:cycle}, we can choose a permutation $\pi$ of the vertices and restrict the representations to be $\pi$-ordered.
  Further, it is sufficient to show hardness for the stretchability problems, as the problems can only become easier with additional information.
  
  We first show $\exists\R$-membership. 
  Note that each of the straight-line objects we consider can be represented with at most four variables:
  For segments, we use two variables for each endpoint. 
  For rays, we use two variables for the apex and two variables for the direction. 
  The condition that two objects intersect can be formulated with constant-degree polynomials in those variables. 
  Hence, each of the problems can be formulated as a sentence in the first-order theory of the reals of the desired form.

  Let us now turn our attention to the $\exists\R$-hardness.
  We will reduce from stretchability of pseudoline arrangements. 
  Given a be a pseudoline arrangement $\L$, we will construct a graph
  $G_\L$  and a permutation $\pi$ such that:
  \begin{enumerate}
   \item \label{itm:Dir1}
   If $\L$ is stretchable then $G_\L$ has a $\pi$-ordered representation with \emph{grounded segments}.
   \item \label{itm:Dir2}
	   If $\L$ is not stretchable then there does not exist a $\pi$-ordered representation of $G_\L$ as an \emph{outer segment graph}.\\
  \end{enumerate}
  Recall that we know the following relations for the considered graph classes. 
  \[\textrm{grounded segments}
  \quad \subseteq \quad  
  \textrm{rays} 
  \quad  \subseteq \quad 
  \textrm{outer segments}.\]
  Thus, Item~\ref{itm:Dir1} implies that $G_\L$ has a $\pi$-ordered representation with rays or outer segments.
  Furthermore, Item~\ref{itm:Dir2} implies that $G_\L$ has no $\pi$-ordered representation with rays or grounded segments.

  We start with the construction of $G_\L$ and $\pi$.
  Let \L be an arrangement of $n$ pseudolines.
  Recall that we can represent \L by $x$-monotone curves.
  Let $\gline_1$ and $\gline_2$ be two vertical lines such that all the intersections of \L lie between $\gline_1$ and $\gline_2$.
  We cut away the part outside the strip bounded by $\gline_1$ and $\gline_2$.
  This gives us a $\pi$-ordered grounded $1$-string representation $R_\L$ with respect to the grounding line $\gline_1$.

  \begin{figure}[htbp]
  \centering
  \includegraphics{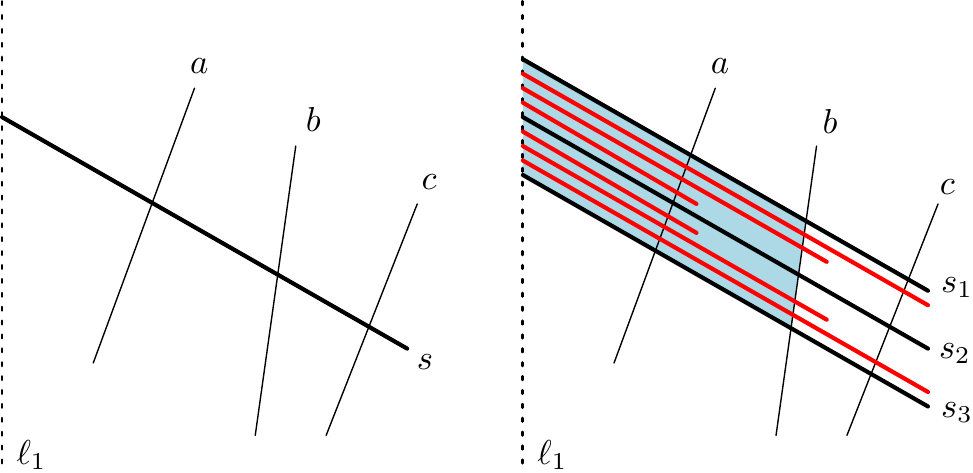}
  \caption{Illustration of \Cref{thm:RecHardness}: Construction of $G_\L$ and its grounded 1-string representation.}
  \label{fig:Stretchability}
\end{figure}
  
  Now we replace each string $s$ representing a pseudoline in \L by the following construction (extending $\pi$ accordingly):
  We split $s$ into three similar copies $s_1,s_2,s_3$, shifted vertically by an offset that is chosen sufficiently small so that the three copies intersect the other pseudolines (and their shifted copies) in the same order.
  For each successive intersection point of $s$ with a pseudoline $s'$ in \L, we add a pair of strings grounded on either side of the base point of $s_2$ and between the base points of $s_1$ and~$s_3$, intersecting none of $s_1$, $s_2$ and~$s_3$.
  The two strings intersect all the pseudolines of \L that $s$ intersects, up to and including $s'$, in the same order as $s$ does. 
  All the strings for $s$ are pairwise nonintersecting and nested around $s_2$; see Figure~\ref{fig:Stretchability}.
  We refer to these pairs of strings as {\em probes}. 
  The probes are meant to enforce the order of the intersections in all $\pi$-ordered representations.

  We now prove Item~\ref{itm:Dir1}.
  We suppose there is a straight line representation of \L, which we denote by \K.
  Again let $\gline_1$ and $\gline_2$ be two vertical lines such that all intersections of \K are contained in the vertical strip between them.
  This gives us a collection of grounded segments $R_\K$. 
  One can check that the above construction involving probes can be implemented using straight line segments, just as illustrated in Figure~\ref{fig:Stretchability}. 
  Thus, $R_\K$ is a $\pi$-ordered grounded segment representation of $G_\L$, as claimed. 
  
  Next, we turn our attention to Item~\ref{itm:Dir2} and suppose that \L is not stretchable.
  Let us further suppose, for the purpose of contradiction, that we have a $\pi$-ordered outer segment representation of $G_\L$.
  We show that keeping only the middle copy $s_2$ of each segment $s$ representing a pseudoline of \L in our construction, we obtain a realization of \L with straight lines.
  For this, we need to prove that the construction of the probes indeed forces the order of the intersections.
  We consider each such segment $s_2$ and orient it from its base point to its other endpoint.
  Now suppose that there exist strings $a$ and $b$ such that the order of intersections of $s_2$ with $a$ and $b$ with respect to this orientation does not agree with that of the pseudoline arrangement.
  (In Figure~\ref{fig:Stretchability}, suppose that $s_2$ crosses the lines $b$ before $a$ in the left-to-right order.)
  We consider the convex region bounded by the arc of the grounding circle between the base points of $s_1$ and $s_3$, and segments from $s_1$, $b$, and $s_3$.
  This convex region is split into two convex boxes by $s_2$. 
  The pair of probes corresponding to the intersection of $s_2$ and $a$ is completely contained in this region, with one probe in each box.
  But now the line $a$ must enter both boxes, thereby intersecting $s_2$ on the left of $b$ with respect to the chosen orientation, a contradiction.
  Therefore, the order of the intersections is preserved, and the collection of segments~$s_2$ is a straight line realization of \L, a contradiction to the assumption that \L is not stretchable.
\end{proof}
  
As there exist pseudoline arrangements, which are not stretchable, we conclude that 
outer segment graphs are a \emph{proper} subclass of outer string graphs.

\section{Rays and Segments} \label{sec:raySeg}

\begin{figure}[htbp]
  \centering
  \includegraphics{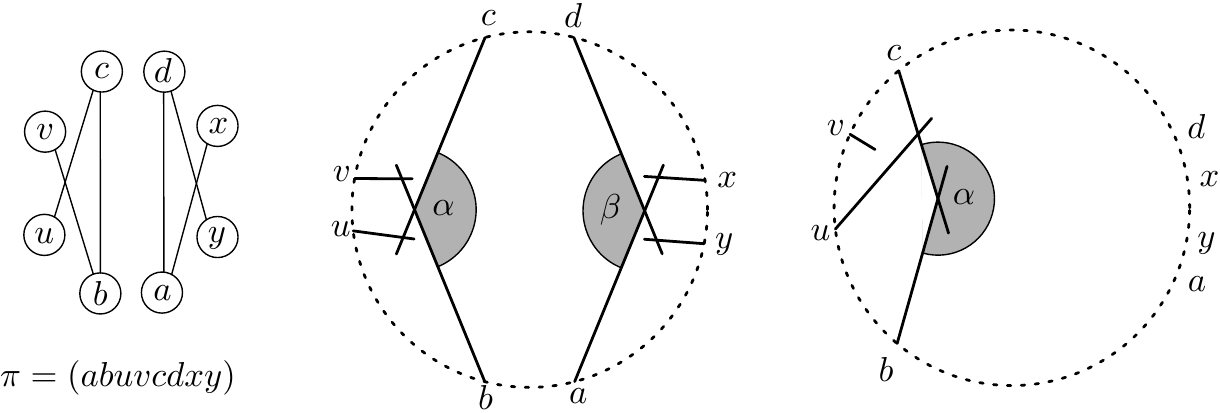}
  \caption{Illustration of \Cref{thm:RayNotOuter}. On the left is a graph $G$ together with a permutation $\pi$ of the vertices displayed. In the middle is a $\pi$-ordered outer segment representation of $G$. The right drawing illustrates that the angles $\alpha$ and $\beta$ must each be at most $180^\circ$.}
  \label{fig:RayNotOuter}
\end{figure}
\begin{theorem}[Rays $\subsetneq$ Outer Segments]\label{thm:RayNotOuter}
  There are graphs that admit a representation as outer segment graphs but not as ray graphs.
\end{theorem}
\begin{proof}
    Consider the graph $G$ and a permutation $\pi$ as displayed in \Cref{fig:RayNotOuter}.
    We will show that $G$ has a $\pi$-ordered representation as an outer segment graph, but not as a ray graph. 
    This implies that $G^\pi$ has a representation as an outer segment graph, but not as a ray graph as well, see Lemma~\ref{lem:cycle}. 
%
%

    Given any $\pi$-ordered representation of $G$, we define the angle $\alpha$ as the angle at the intersection of $b$ and $c$ towards the segments $d,x,y,a$ and we define the angle $\beta$ as the angle at the intersection of $a$ and $d$ towards the segments $b,u,v,c$, as can be seen in \Cref{fig:RayNotOuter}.
    
    We show that both $\alpha$ and $\beta$ are smaller than $180^\circ$ in any outer segment representation.
    As the two cases are symmetric we show it only for $\alpha$.
    Assume $\alpha \geq 180^\circ$ as on the right of \Cref{fig:RayNotOuter}.
    If $u$ intersects $c$ (as shown in the figure) then it blocks $v$ from intersecting $b$, as $v$ must not intersect $u$.
    Likewise, $v$ intersecting $b$ would block $u$ from intersecting $c$.
    This shows $\alpha,\beta < 180^\circ$.
    
    As the angles  are smaller than $180^\circ$, we conclude that either the extensions of $a$ and~$b$ or the extensions of $c$ and $d$ must meet outside of the grounding circle. 
    Recall that we considered any representation of $G$. 
    By Lemma~\ref{lem:Raycharacterization} it holds for every ray graph that there exists at least one representation of $G$ with outer segments such that \emph{all} extensions meet \emph{within} the grounding circle. 
    (The lemma also holds for ordered representations.)
    Thus there cannot be a $\pi$-ordered ray representation of $G$.
\end{proof}

\begin{figure}[htbp]
  \centering
  \includegraphics{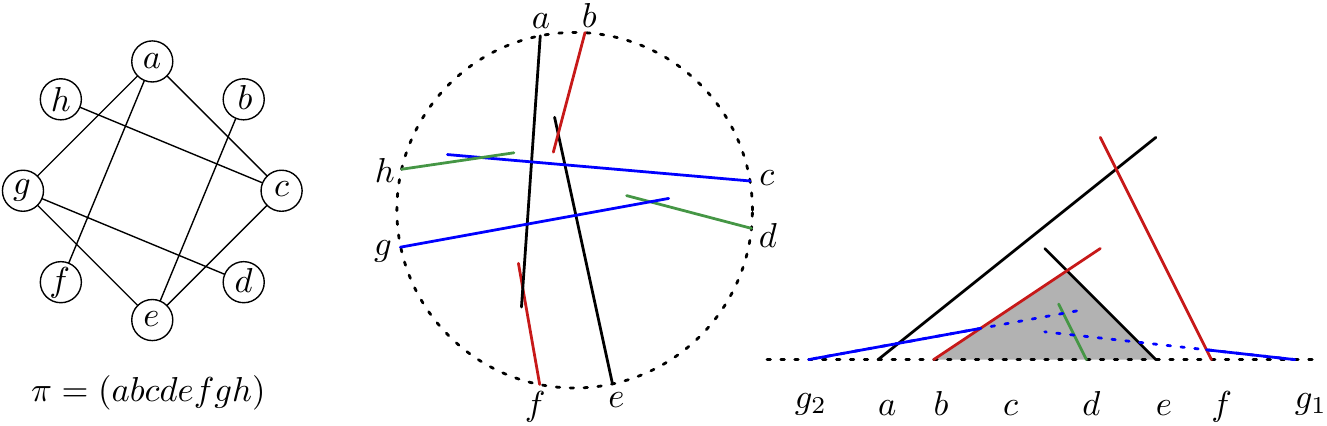}
  \caption{Illustration of \Cref{thm:RayNotDownwardRay}: A graph $G$ together with a permutation $\pi$ of the vertices (left); A $\pi$-ordered outer segment representation of $G$ (middle); The segment $g$ cannot enter the gray triangle without intersecting $b$ or $f$ (right).}
  \label{fig:RayNotDownwardRay}
\end{figure}

\begin{theorem}[Downward Rays $\subsetneq$ Rays]\label{thm:RayNotDownwardRay}
  There are graphs that admit a representation as ray graphs but not as downward ray graphs.
\end{theorem}

\begin{proof}
	Consider the graph $G$ and the permutation $\pi$ as displayed in \Cref{fig:RayNotDownwardRay}~(left).
	Clearly, $G$ has a $\pi$-ordered representation as a ray graph, as can be seen from the outer segment reperesentation of $G$ in shown in \Cref{fig:RayNotDownwardRay}~(middle).
    We will show that $G$ does not have a $\pi$-ordered representation as a grounded segment graph. 
	Hence $G^\pi$ has a representation as a ray graph, but not as a grounded segment graph or a downward ray graph; see Lemma~\ref{lem:cycle} and Lemma~\ref{lem:DRayEqualGSegment}. 
%
%
	
	Assume for the sake of contradiction that $G$ has a $\pi$-ordered representation $R^\pi$ as grounded segment graph.
	As $G$ is rotation symmetric, we can assume without loss of generality that in~$R^\pi$ the base points of the segments $a$, $b$, $c$, $d$, $e$, and $f$ are sorted from left to right in this order along $\gline$; cf.\ \Cref{fig:RayNotDownwardRay}~(right).
	Consider first the segments $a$, $f$, $b$ and~$e$ in~$R^\pi$.  As $a$ and $f$ intersect, they form a triangle~$\Delta_{af}$ together with $\gline$. 
	An according statement holds for $b$ and $e$ with triangle~$\Delta_{be}$. Moreover, as none of $a$ and $f$ intersects $b$ or $e$, and as the base points of $b$ and $e$ lie between the base points of $a$ and $f$, the triangle $\Delta_{be}$, as well as the whole segments $b$ and $e$ lie completely inside $\Delta_{af}$.
    Now consider the segment $d$, which has its base point between $b$ and $e$. As $d$ does not intersect any of $b$ and $e$, $d$ lies completely inside $\Delta_{be}$.
	Finally, consider the segment $g$ which has its base point either to the left of $a$ or to the right of $f$.
	The two possibilities are indicated with $g_1$ and $g_2$ in \Cref{fig:RayNotDownwardRay}~(right).
	On the one hand, $g$ must intersect $d$ and hence enter the triangle $\Delta_{be}$. On the other hand, $g$ is not allowed to intersect any of $b$ and $f$, a contradiction.
\end{proof}

\section*{Acknowledgments}

This work was initiated during the Order \& Geometry Workshop organized by Piotr Micek and the second author at the Gułtowy Palace near Poznań, Poland, on September 14-17, 2016. We thank the organizers and attendees, who contributed to an excellent work atmosphere. Some of the problems tackled in this paper were brought to our attention during the workshop by Michał Lasoń. The first author also thanks Sergio Cabello for insightful discussions on these topics.

\bibliographystyle{abbrv}
\bibliography{main}

\begin{thebibliography}{10}

\bibitem{benzer59}
S.~Benzer.
\newblock On the topology of the genetic fine structure.
\newblock {\em Proceedings of the National Academy of Sciences},
  45(11):1607--1620, 1959.

\bibitem{CCL13}
S.~Cabello, J.~Cardinal, and S.~Langerman.
\newblock The clique problem in ray intersection graphs.
\newblock {\em Discrete {\&} Computational Geometry}, 50(3):771--783, 2013.

\bibitem{CJ16}
S.~{Cabello} and M.~{Jej{\v c}i{\v c}}.
\newblock {Refining the Hierarchies of Classes of Geometric Intersection
  Graphs}.
\newblock {\em arXiv}, 1603.08974, 2016.

\bibitem{CJ16abs}
S.~Cabello and M.~Jej\v{c}i\v{c}.
\newblock Refining the hierarchies of classes of geometric intersection graphs.
\newblock {\em Electronic Notes in Discrete Mathematics}, 54:223--228, 2016.

\bibitem{C88}
J.~F. Canny.
\newblock Some algebraic and geometric computations in {PSPACE}.
\newblock In {\em Proc. STOC}, pages 460--467. ACM, 1988.

\bibitem{C15}
J.~Cardinal.
\newblock Computational geometry column {62}.
\newblock {\em ACM SIGACT News}, 46(4):69--78, 2015.

\bibitem{CG09}
J.~Chalopin and D.~Gon{\c{c}}alves.
\newblock Every planar graph is the intersection graph of segments in the
  plane: extended abstract.
\newblock In {\em Proc. STOC}, pages 631--638. ACM, 2009.

\bibitem{CFHW15}
S.~{Chaplick}, S.~{Felsner}, U.~{Hoffmann}, and V.~{Wiechert}.
\newblock {Grid Intersection Graphs and Order Dimension}.
\newblock {\em arXiv}, 1512.02482, 2015.

\bibitem{CHOSU14}
S.~Chaplick, P.~Hell, Y.~Otachi, T.~Saitoh, and R.~Uehara.
\newblock Intersection dimension of bipartite graphs.
\newblock In {\em Proc. TAMC}, volume 8402 of {\em LNCS}, pages 323--340.
  Springer, 2014.

\bibitem{EET76}
G.~Ehrlich, S.~Even, and R.~E. Tarjan.
\newblock Intersection graphs of curves in the plane.
\newblock {\em J. Comb. Theory, Ser. {B}}, 21(1):8--20, 1976.

\bibitem{F14}
S.~Felsner.
\newblock The order dimension of planar maps revisited.
\newblock {\em {SIAM} J. Discrete Math.}, 28(3):1093--1101, 2014.

\bibitem{KMPV17}
J.~M. Keil, J.~S.~B. Mitchell, D.~Pradhan, and M.~Vatshelle.
\newblock An algorithm for the maximum weight independent set problem on
  outerstring graphs.
\newblock {\em Comput. Geom.}, 60:19--25, 2017.

\bibitem{KN98}
A.~V. Kostochka and J.~Ne{\v s}et{\v r}il.
\newblock Coloring relatives of intervals on the plane, {I:} chromatic number
  versus girth.
\newblock {\em Eur. J. Comb.}, 19(1):103--110, 1998.

\bibitem{KN02}
A.~V. Kostochka and J.~Ne{\v s}et{\v r}il.
\newblock Colouring relatives of intervals on the plane, {II:} intervals and
  rays in two directions.
\newblock {\em Eur. J. Comb.}, 23(1):37--41, 2002.

\bibitem{K91}
J.~Kratochv{\'{\i}}l.
\newblock String graphs. {I.} the number of critical nonstring graphs is
  infinite.
\newblock {\em J. Comb. Theory, Ser. {B}}, 52(1):53--66, 1991.

\bibitem{K91a}
J.~Kratochv{\'{\i}}l.
\newblock String graphs. {II.} recognizing string graphs is {NP}-hard.
\newblock {\em J. Comb. Theory, Ser. {B}}, 52(1):67--78, 1991.

\bibitem{KM91}
J.~Kratochv{\'{\i}}l and J.~Matou{\v{s}}ek.
\newblock String graphs requiring exponential representations.
\newblock {\em J. Comb. Theory, Ser. {B}}, 53(1):1--4, 1991.

\bibitem{KM94}
J.~Kratochv{\'{\i}}l and J.~Matou{\v{s}}ek.
\newblock Intersection graphs of segments.
\newblock {\em J. Comb. Theory, Ser. {B}}, 62(2):289--315, 1994.

\bibitem{matouvsek2002lectures}
J.~Matou{\v{s}}ek.
\newblock {\em Lectures on discrete geometry}, volume 108.
\newblock Springer, 2002.

\bibitem{M14}
J.~Matou{\v{s}}ek.
\newblock Intersection graphs of segments and $\exists\mathbb{R}$.
\newblock {\em arXiv}, 1406.2636, 2014.

\bibitem{MM99}
T.~A. McKee and F.~McMorris.
\newblock {\em Topics in Intersection Graph Theory}.
\newblock Society for Industrial and Applied Mathematics, 1999.

\bibitem{MNTTU16}
I.~Musta\c{t}\u{a}, K.~Nishikawa, A.~Takaoka, S.~Tayu, and S.~Ueno.
\newblock On orthogonal ray trees.
\newblock {\em Discrete Applied Mathematics}, 201:201--212, 2016.

\bibitem{N85}
W.~Naji.
\newblock Reconnaissance des graphes de cordes.
\newblock {\em Discrete Mathematics}, 54(3):329 -- 337, 1985.

\bibitem{RW14}
A.~Rok and B.~Walczak.
\newblock Outerstring graphs are {\(\chi\)}-bounded.
\newblock In {\em Proc. SoCG}, page 136. ACM, 2014.

\bibitem{S09}
M.~Schaefer.
\newblock Complexity of some geometric and topological problems.
\newblock In {\em Proc. GD}, volume 5849 of {\em LNCS}, pages 334--344.
  Springer, 2009.

\bibitem{SSS03}
M.~Schaefer, E.~Sedgwick, and D.~Stefankovic.
\newblock Recognizing string graphs in {NP}.
\newblock {\em J. Comput. Syst. Sci.}, 67(2):365--380, 2003.

\bibitem{S90}
P.~W. Shor.
\newblock Stretchability of pseudolines is np-hard.
\newblock In {\em Applied Geometry And Discrete Mathematics}, volume~4 of {\em
  DIMACS Series in DMTCS}, pages 531--554. AMS, 1990.

\bibitem{STU10}
A.~M.~S. Shrestha, S.~Tayu, and S.~Ueno.
\newblock On orthogonal ray graphs.
\newblock {\em Discrete Applied Mathematics}, 158(15):1650--1659, 2010.

\bibitem{S66}
F.~W. Sinden.
\newblock Topology of thin film {RC} circuits.
\newblock {\em Bell System Technical Journal}, 45(9):1639--1662, 1966.

\bibitem{ST11}
J.~A. Soto and C.~Telha.
\newblock Jump number of two-directional orthogonal ray graphs.
\newblock In {\em Proc. {IPCO}}, volume 6655 of {\em LNCS}, pages 389--403.
  Springer, 2011.

\bibitem{WP85}
W.~Wessel and R.~P\"oschel.
\newblock On circle graphs.
\newblock In H.~Sachs, editor, {\em Graphs, Hypergraphs and Applications},
  pages 207--210. Teubner, 1985.

\end{thebibliography}
%
%

\end{document}